\documentclass[submission,copyright,creativecommons]{eptcs}
\providecommand{\event}{QPL 2020} 
\usepackage{underscore}           
\usepackage{enumitem}
\usepackage{amsthm}

\theoremstyle{definition}
\newtheorem{problemthm}{Problem}[section]

\makeatletter
\newenvironment{problem}[1][]{%
  \begin{problemthm} #1%
  \begin{description}[%
    labelwidth=3.5em,
    labelindent=1em,
    itemsep=0.25ex,
    font=\mdseries\itshape
  ]%
}{
  \end{description}%
  \end{problemthm}%
}
\makeatother

\newlist{algenum}{enumerate}{9}
\setlist[algenum,1]{%
  label=\arabic*.,
  itemsep=0ex,
  topsep=1ex,
  leftmargin=2em}
\setlist[algenum,2]{%
  label=\alph*\upshape.,
  itemsep=0ex,
  topsep=-0.375ex,
  leftmargin=1.5em}
\setlist[algenum,3]{%
  label=(\!\!\;\itshape\roman*\;\!\upshape),
  topsep=0ex,
  itemsep=0.375ex,
  leftmargin=2.5em}

\usepackage{mathtools}

\newcommand{\bw}{1}

\usepackage[margin=3cm]{geometry}

\usepackage{hyperref}

\hypersetup{
     colorlinks = true,
     linkcolor = blue!70!black,
     anchorcolor = blue!70!black,
     citecolor = red!70!black,
     filecolor = blue!70!black,
     urlcolor = blue!70!black
}

\usepackage{amsmath,amsthm,amssymb}
\usepackage{xspace,enumerate,color,epsfig}
\usepackage{graphicx}
\graphicspath{{.}{./figures/}}
\usepackage{marvosym}
\usepackage{bm}
\usepackage{tabu}
\usepackage{tikzit}
\usetikzlibrary{matrix}
\usetikzlibrary{fit}
\usetikzlibrary{calc}

\if\bw1

\tikzstyle{dot}=[inner sep=0.3mm, minimum width=2mm, minimum height=2mm, draw, shape=circle, font={\footnotesize}, tikzit fill=magenta]
\tikzstyle{white dot}=[dot, fill=white, text depth=-0.2mm, tikzit category=ZH-pf]
\tikzstyle{gray dot}=[dot, fill={rgb,255: red,170; green,170; blue,170}, text depth=-0.2mm, tikzit category=ZH-pf]
\tikzstyle{H box}=[rectangle, fill=white, draw=black, xscale=1, yscale=1, font={\footnotesize}, inner sep=0.75pt, minimum width=0.3cm, minimum height=0.3cm, scale=0.75, tikzit shape=rectangle]
\tikzstyle{ug}=[regular polygon, regular polygon sides=3, fill={zx_red}, draw=black, inner sep=0pt, minimum width=1em, tikzit draw=blue]
\tikzstyle{st}=[star, star points=5, fill=white, draw=black, inner sep=1.2pt, line width=1.2pt, tikzit fill=blue, tikzit draw=red, tikzit category=ZH-pf]
\tikzstyle{triangle}=[regular polygon, regular polygon sides=3, fill=white, draw=black, inner sep=0pt, minimum width=1em, tikzit draw=blue, tikzit category=ZH-pf]
\tikzstyle{not}=[fill={rgb,255: red,170; green,170; blue,170}, draw=black, shape=circle, dot, minimum width=2.5mm, label={center:\tiny$\bm\neg$}]
\tikzstyle{bbindex}=[font={\color{blue}\footnotesize}]
\tikzstyle{wide point}=[fill=white, draw, shape=isosceles triangle, shape border rotate=-90, isosceles triangle stretches=true, inner sep=0pt, minimum width=1.5cm, minimum height=6.12mm, yshift=-0.0mm]
\tikzstyle{medium gray box}=[semilarge box, fill={gray!30}]
\tikzstyle{small box}=[rectangle, inline text, fill=white, draw, minimum height=5mm, yshift=-0.5mm, minimum width=5mm, font={\small}]
\tikzstyle{small gray box}=[small box, fill={gray!30}]
\tikzstyle{medium box}=[rectangle, inline text, fill=white, draw, minimum height=5mm, yshift=-0.5mm, minimum width=8mm, font={\small}]

\tikzstyle{gray}=[-, draw={blue!60!white}, tikzit draw=blue]
\tikzstyle{brace edge}=[-, tikzit draw=blue, decorate, decoration={brace,amplitude=1mm,raise=-1mm}]
\tikzstyle{diredge}=[->]
\tikzstyle{bbox edge}=[-, draw={rgb,255: red,42; green,145; blue,255}]

\input{zh.tikzdefs}
\fi

\if\bw0

\input{zh.tikzdefs}
\fi

\usepackage{stmaryrd}
\usepackage{docmute}
\usepackage{keycommand}

\theoremstyle{definition}
\newtheorem{theorem}{Theorem}[section]

\newtheorem{definition}[theorem]{Definition}

\newtheorem{example}[theorem]{Example}

\newtheorem{example*}[theorem]{Example*}
\newtheorem{examples*}[theorem]{Examples*}

\newtheorem{remark*}[theorem]{Remark*}

\usepackage[color]{changebar}

\usepackage{color}
\def\bR{\begin{color}{red}} 
\def\bB{\begin{color}{blue}}
\def\bM{\begin{color}{magenta}}
\def\bC{\begin{color}{cyan}}
\def\bW{\begin{color}{white}}
\def\bBl{\begin{color}{black}} 
\def\bG{\begin{color}{green}}
\def\bY{\begin{color}{yellow}}
\def\e{\end{color}\xspace}
\newcommand{\bit}{\begin{itemize}}
\newcommand{\eit}{\end{itemize}\par\noindent}
\newcommand{\ben}{\begin{enumerate}}
\newcommand{\een}{\end{enumerate}\par\noindent}
\newcommand{\beq}{\begin{equation}}
\newcommand{\eeq}{\end{equation}\par\noindent}
\newcommand{\beqa}{\begin{eqnarray*}}
\newcommand{\eeqa}{\end{eqnarray*}\par\noindent}
\newcommand{\beqn}{\begin{eqnarray}}
\newcommand{\eeqn}{\end{eqnarray}\par\noindent}

\newcommand{\bra}[1]{\ensuremath{\langle #1 |}}
\newcommand{\ket}[1]{\ensuremath{|  #1 \rangle}}

\newcommand{\ketbra}[2]{\ensuremath{\ket{#1}\!\bra{#2}}}
\newcommand{\intf}[1]{\left\llbracket #1 \right\rrbracket} 

\tikzstyle{inline text}=[text height=1.5ex, text depth=0.25ex,yshift=0.5mm]
\tikzstyle{semilarge box}=[rectangle,inline text,fill=white,draw,minimum height=5mm,yshift=-0.5mm,minimum width=12.5mm,font=\small]
\tikzstyle{hadamard}=[fill=white,draw,inner sep=0.6mm,minimum height=1.5mm,minimum width=1.5mm]
\tikzstyle{small hadamard}=[fill=white,draw,inner sep=0.6mm,minimum height=1.5mm,minimum width=1.5mm]

\newcommand\gendiagram[1]{%
\begin{tikzpicture}
        \begin{pgfonlayer}{nodelayer}
                \node [style=none] (0) at (-0.75, 1) {};
                \node [style=none] (1) at (0.75, 1) {};
                \node [style=none] (2) at (-0.75, -1) {};
                \node [style=none] (3) at (0.75, -1) {};
                \node [style=none] (4) at (0, 0.75) {$\ldots$};
                \node [style=semilarge box] (5) at (0, -0) {#1};
                \node [style=none] (6) at (0, -0.75) {$\ldots$};
        \end{pgfonlayer}
        \begin{pgfonlayer}{edgelayer}
                \draw (0.center) to (2.center);
                \draw (1.center) to (3.center);
        \end{pgfonlayer}
\end{tikzpicture}
}
\newcommand{\dotmult}[1]{%
\,\begin{tikzpicture}
        \node [#1] (a) {};
        \draw (a) -- (90:0.55);
        \draw (a) (-45:0.6) -- (a);
        \draw (a) (-135:0.6) -- (a);
\end{tikzpicture}\,\xspace}
\newcommand{\graymult}{\dotmult{gray dot}}
\newcommand{\grayphase}[1]{\phase[style=gray phase dot]{#1}}
\newkeycommand{\phase}[style=white phase dot][1]{\,\begin{tikzpicture}
    \begin{pgfonlayer}{nodelayer}
        \node [style=none] (0) at (0, 0.6) {};
        \node [style=\commandkey{style}] (2) at (0, -0) {$#1$};
        \node [style=none] (3) at (0, -0.6) {};
    \end{pgfonlayer}
    \begin{pgfonlayer}{edgelayer}
        \draw (2) to (0.center);
        \draw (3.center) to (2);
    \end{pgfonlayer}
\end{tikzpicture}\,}
\tikzstyle{gray phase dot}=[gray dot]

\newcommand\SAT{\textrm{\upshape SAT}}

\usepackage[utf8]{inputenc}

\title{Tensor Network Rewriting Strategies\\
for Satisfiability and Counting}

\author{Niel de Beaudrap {\footnotesize$^1$}, Aleks Kissinger {\footnotesize$^1$}, Konstantinos Meichanetzidis {\footnotesize$^{1,2}$}
\institute{$^1$ Quantum Group, Department of Computer Science,
University of Oxford}
\institute{$^2$ Cambridge Quantum Computing Ltd.}
\email{\{firstname.lastname\}@cs.ox.ac.uk}
}

\def\titlerunning{Tensor Network Rewriting Strategies for Satisfiability and Counting}
\def\authorrunning{N. de Beaudrap, A. Kissinger \& K. Meichanetzidis}

\begin{document}

\maketitle

\begin{abstract}
\vspace*{-1ex}
We provide a graphical treatment of SAT and \#SAT on equal footing.
Instances of \#SAT can be represented as tensor networks in a standard way.
These tensor networks are interpreted by diagrams of the ZH-calculus: a system to reason about tensors over $\mathbb{C}$ in terms of diagrams built from simple generators,
in which computation may be carried out by \emph{transformations of diagrams alone}.
In general, nodes of ZH diagrams take parameters over $\mathbb{C}$ which determine the tensor coefficients;
for the standard representation of \#SAT instances,
the coefficients take the value $0$ or $1$.
Then, by choosing the coefficients of a diagram to range over
$\mathbb B$, we represent the corresponding instance of SAT.
Thus, by interpreting a diagram either over the boolean semiring
or the complex numbers,
we instantiate either the \emph{decision} or \emph{counting} version of the problem.
We find that for classes known to be in P, such as $2$SAT and \#XORSAT,
the existence of appropriate rewrite rules allows for efficient simplification of the diagram, producing the solution in polynomial time.
In contrast, for classes known to be NP-complete, such as $3$SAT,
or \#P-complete, such as \#$2$SAT, the corresponding rewrite rules
introduce hyperedges to the diagrams, in numbers which are not easily bounded above by a polynomial.
This diagrammatic approach unifies the diagnosis of the complexity of
CSPs and \#CSPs and shows promise in aiding tensor network contraction-based algorithms.

\end{abstract}

{\footnotesize
\hspace*{\fill}{\sl The solution lies within the problem\\
\hspace*{\fill}The answer is in every question} \\
\hspace*{\fill} - Funkadelic
} 

\section{Introduction}

The Boolean satisfiability problem, or SAT,
and its variants,
is central to theoretical computer science and complexity theory,
with many practical real-world applications.
The formula defining the problem
can always be given in conjunctive normal form (CNF)
comprising of an AND of OR-\emph{constraints}.
In a boolean formula on variables $(x_1, x_2, \ldots, x_n)$, a \emph{literal} $\ell$ is either some variable $x_a$ or its negation $\neg x_a$, for some $1 \le a \le n$.

\begin{problem}[CNFSAT (or simply SAT)]
\label{def:SAT}
\item[Input:]
  $
    \smash{
      \phi(x)
    =
        \bigwedge\limits_{i=1}^m C_i(x)
    }
  $,
  where $C_i(x) = \smash{\bigvee\limits_{j = 1}^{k_i} \ell_{i,j}}$.
  \vspace{.2cm}
\item[Output:]
  $\exists x \in \{0,1\}^n : \phi(x)$.
\end{problem}

\noindent
The formula consists of a conjunction of clauses $C_i(x)$ for $1 \le i \le m$, each being a disjunction of a $k_i \in \mathbb N$ number of literals.
SAT is the \emph{decision} problem of determining whether there \emph{exists} an assignment $x$ of the variables satisfying \emph{all} clauses.
Given a formula $\phi$ in CNF form, we write $\exists \phi$ for the \emph{instance of \SAT\ to evaluate} ${\exists x \!:\! \phi(x)}$, and $[\exists \phi]$ for its answer.

Famously, \SAT\ is NP-complete.
At the heart of this work is the family of special cases $k$SAT,
where $k_i \le k$ for every clause $C_i$.
For $k=2$ we have that $2$SAT $\in$ P.
However, $3$SAT $\in$ NP-complete, and so are all cases for $k > 3$~\cite{10.5555/2086753,cook,levin}.
An important special case of SAT (in the more general formulation of determining whether a boolean formula of \emph{some} form is satisfiable) is XORSAT,
where the OR clauses are replaced with XORs,
\emph{i.e.}~replacing $\vee \to \oplus$ in Def.~\ref{def:SAT}.
This special case is significantly easier, in that XORSAT $\in$ P.

For each decision problem, $\mathrm{X}$,
the corresponding \emph{counting} problem
$\#\mathrm{X}$ asks for the \emph{number} of satisfying assignments to a Boolean formula.
In particular, \#SAT is given as follows:\\

\begin{problem}[\#CNFSAT (or simply \#SAT)] 
\label{def:sharpSAT}
\item[Input:] 
  $
    \smash{
      \phi(x)
    =
        \bigwedge\limits_{i=1}^m C_i(x)
    }
  $,
  where $C_i(x) = \smash{\bigvee\limits_{j = 1}^{k_i} \ell_{i,j}}$.
  \vspace{.2cm}
\item[Output:]
  $\# \bigl\{ x \in \{0,1\}^n : \phi(x)  \}$.
\end{problem}

In general, \#SAT $\in$ \#P-complete~\cite{valiant1979complexity}.
The counting versions of the special cases defined above are denoted similarly \#XORSAT, \#$k$SAT.
Given a formula $\phi$ in CNF form, we write $\# \phi$ for the \emph{instance of \#\SAT} described by $\phi$, and $[\#\phi] := \#\{x: \phi(x)\}$ for its solution.
Note that 
\#XORSAT $\in$ P and that \#$k$SAT $\in$ \#P-complete for $k \geq 2$.
In this sense, \#$2$SAT is significantly harder than its decision version.

This work is motivated by the will to understand the structural origin of the easiness or hardness of these problems.
What is it that makes counting the solutions to XORSAT, and consequently deciding if there is at least one, easy?
Why is $2$SAT easy to decide but hard to count?
What changes when we turn to $3$SAT
and both deciding and counting are hard?
In the following, we investigate these questions by employing graphical methods.

\section{Counting with ZH rewrites}

Instances of \#\SAT\ may be represented as tensor networks~\cite{10.21468/SciPostPhys.7.5.060,garcia-saez2011} and the solution is returned by full tensor contraction, a \#P-hard problem in general~\cite{Damm2002}.
Here,
we describe problem instances in terms of a
\emph{specific variety} of tensor network, i.e. diagrams of the ZH-calculus~\cite{backens2018zh,wetering2019completeness} --- a formal system to analyse complex-valued tensor networks of bond-dimension $2$
(all indices ranging over $\{0,1\}$)
by transformations of diagrams.
This provides a means of considering the complexity of these problems by tensor-based techniques, while nevertheless avoiding explicit tensor contractions.

\vspace*{-1ex}
\subsection{Tensor networks from the ZH calculus}
\label{sec:ZH-intro}

The ZH calculus is a diagrammatic notation for tensor networks, which allows computations to be done with these diagrams \emph{in lieu of matrix computations}.
These diagrams consist of graphs with different kinds of nodes:
  \emph{Z-spiders}, which are represented as \emph{white dots} with a specific number of wires,
  and
  \emph{H-boxes}, which are represented as white boxes labelled with a complex number $\alpha\in\mathbb{C}$.
These generators are interpreted as follows, where $\intf{\,\cdot\,}$ denotes the map from diagrams to matrices:
\vspace*{-1ex}
\begin{equation*}
  \label{eqn:ZH-generators}
    \intf{\input{Z-spider-big.tikz}}
  :=
    \ket{0}^{\otimes n}\bra{0}^{\otimes m} + \ket{1}^{\otimes n}\bra{1}^{\otimes m}
    \;,
\qquad\qquad
    \intf{\input{H-spider.tikz}}
  \;:=\;\;
    \sum_{\mathclap{\substack{
      x \in \{0,1\}^n \\ y \in \{0,1\}^m
    }}}
    \;
      \alpha^{x_1\ldots x_n y_1\ldots y_m}
      \ket{x}\bra{y}  
\end{equation*}~\\[-2ex]
The $\mathrm{Z}$-spider is a high-dimensional Kronecker delta-function whose matrix elements are all zeros except
for the two entries indexed by all zeros or all ones which are $1$.
An $\mathrm{H}(\alpha)$-box represents an all-ones matrix with the exception of the entry indexed by all ones which is set to $\alpha$.
If not denoted, by convention $\alpha = -1$.
Straight and curved wires have the interpretations:
\vspace*{-1ex}
\begin{equation*}
    \intf{\;|\big.\;}
  :=
    \ketbra{0}{0}+\ketbra{1}{1}
  \;,
\qquad\qquad
    \intf{\;\begin{tikzpicture}
	\begin{pgfonlayer}{nodelayer}
		\node [style=none] (0) at (-2, 0.25) {};
		\node [style=none] (1) at (-1.5, -0.25) {};
		\node [style=none] (2) at (-1, 0.25) {};
	\end{pgfonlayer}
	\begin{pgfonlayer}{edgelayer}
		\draw [bend right=45, looseness=1.00] (0.center) to (1.center);
		\draw [bend right=45, looseness=1.00] (1.center) to (2.center);
	\end{pgfonlayer}
\end{tikzpicture}\;}
  :=
    \ket{00}+\ket{11}
  \;,
\qquad\qquad
    \intf{\;\input{wire-cap.tikz}\;}
  :=
    \bra{00}+\bra{11}
  \;.
\end{equation*}
Diagram juxtaposition corresponds to tensor product and composition to matrix product:
\begin{equation*}{}
\mspace{-18mu}
    \intf{\gendiagram{\footnotesize $D_1$}\;\gendiagram{\footnotesize $D_2$}}
  :=
    \intf{\gendiagram{\footnotesize $D_1$}} \!\otimes\! \intf{\gendiagram{\footnotesize $D_2$}}
    ,
\qquad
    \intf{\input{sequential-composition.tikz}}
  :=
    \intf{\gendiagram{\footnotesize $D_2$}} \!\circ\! \intf{\gendiagram{\footnotesize $D_1$}}
    .
\end{equation*}
Expressions in the ZH calculus also use two derived generators, called \emph{X-spiders}, represented by a gray dot with a number of wires, and \emph{NOT}.
\vspace*{-1ex}
\begin{equation}
  \label{eq:grey-spider-and-neg-def}
  \input{X-spider-dfn.tikz}
\qquad\qquad\qquad
  \input{negate-dfn.tikz}  
\end{equation}
With these definitions, \graymult\ acts on computational basis states as XOR, i.e. one wire caries the parity of the other wires, and \grayphase{\neg} acts as NOT (or Pauli $\mathrm{X}$). We will write a X-spider labelled by a `$\neg$' to represent an X-spider with a NOT applied to any of its legs. Concretely, the correspond to the even- and odd-parity tensors:
\[
\intf{\input{X-spider.tikz}} \ =\ 
\ \ \ \ 
\sum_{\mathclap{\substack{x_1,...,x_m,y_1,...,y_n\\ \bigoplus x_i \oplus \bigoplus y_j = 0}}} \ \ 
\ket{y_1 \ldots y_n} \bra{x_1 \ldots x_m}
\qquad\qquad
\intf{\input{X-neg-spider.tikz}} \ =\ 
\ \ \ \ 
\sum_{\mathclap{\substack{x_1,...,x_m,y_1,...,y_n\\ \bigoplus x_i \oplus \bigoplus y_j = 1}}} \ \ 
\ket{y_1 \ldots y_n} \bra{x_1 \ldots x_m}
\]
and in the following we will also use the compact notation
\[\input{maybe-neg-spider.tikz} \ =\ \begin{cases}
\input{X-spider.tikz} & \textrm{ if } a = 0 \\[8mm]
\input{X-neg-spider.tikz} & \textrm{ if } a = 1
\end{cases}\]


These nodes suffice to express any tensor of any rank with bond-dimension $2$ over $\mathbb C$, as we now show.
For a vector $\vec b \in \{0,1\}^n$, we define the short-hand notation
\vspace*{-1ex}
\begin{equation*}
  \label{eqn:indexing-box}
    \input{indexing-box.tikz}
  \;=\;
    \left(\grayphase{\neg}\right)^{\!\! 1 - b_1} \!\ldots\; \left(\grayphase{\neg}\right)^{\!\! 1 - b_n}.
\end{equation*}
Then, given such a tensor $A$, with $n$ indices, we may construct the ZH diagram representing it,
called the \emph{ZH normal form for $A$},
using one $H$ box for each of its $2^n$ coefficients $a_{x_1 x_2 \cdots x_n}$:
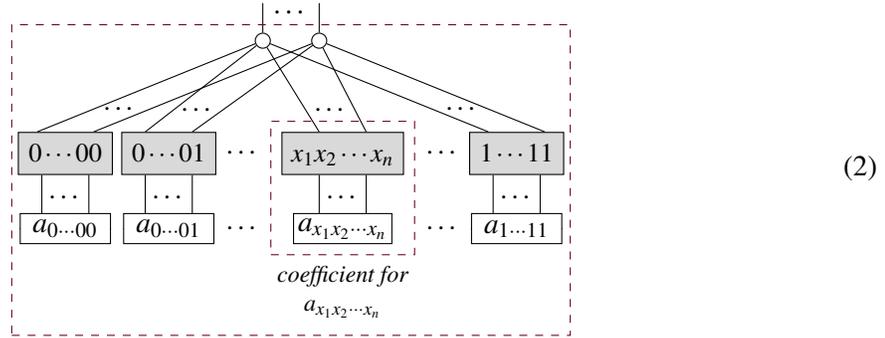
\vspace*{-1ex}
\begin{equation}
\label{eqn:ZHnf}
\begin{gathered}
  \input{ZH-nf-picture.tikz}
\end{gathered} 
\end{equation}
For any coefficient $a_{x_1 x_2 \cdots x_n} = 1$, the corresponding coefficient gadget may be omitted --- we call a ZH diagram with these omitted coefficient gadgets a \emph{condensed normal form}.
These normal forms will in many cases not be the most efficient way to represent a tensor by ZH diagrams; for our purposes it suffices that they exist as a means to represent \emph{any particular} diagram.
Thus, the ZH calculus forms a way to represent tensors with diagrams assembled from simple generators.

\vspace*{-1ex}
\subsection{Representing \#\SAT\ instances as $\mathrm{ZH}$ diagrams}


Instances $\#\phi \in \#\mathrm{SAT}$ may be represented by ZH tensor networks in a straightforward way.

\begin{theorem}\label{thm:zh-sat-rep}
  Any instance $\#\phi \in \#\mathrm{SAT}$
  can be represented as a closed ZH diagram, i.e. with no open wires,
  which evaluates to $[\#\phi]$,
  composed of the following nodes:
  \vspace*{-1ex}
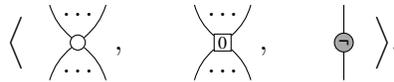
  \begin{equation*}
      \Big\langle~~\input{Z-spider.tikz}~,
    \qquad
      \input{H-zero.tikz}~,
    \qquad
      \begin{tikzpicture}
	\begin{pgfonlayer}{nodelayer}
		\node [style=not] (0) at (0, 0) {};
		\node [style=none] (1) at (0, 1) {};
		\node [style=none] (2) at (0, -1) {};
	\end{pgfonlayer}
	\begin{pgfonlayer}{edgelayer}
		\draw (1.center) to (2.center);
	\end{pgfonlayer}
\end{tikzpicture}
~~\Big\rangle.
  \end{equation*}
\end{theorem}


\vspace*{-2ex}
\begin{proof} \label{pr:cnf2diag}
We construct a bijection between CNF formulae and ZH diagrams,
keeping consistent with the notation in Defs.~\ref{def:SAT}, \ref{def:sharpSAT}.
The construction is of the form of a circuit (preparation, process, measurement) with postselection.

Variables are initialised as ``states'' to all their possible assignments.
Let $m_j$ be the number of clauses in which $x_j$ or its negation participates.
Each $x_j$ is represented by a $m_j$-ary
Z-spider so that it is copied enough times to participate in the clauses,
after possibly going through a negation,
as dictated by $\phi$.
The states then enter the $\mathrm{OR}$-gates,
which play the role of the ``processes''.
Let $k_i$  be the number of variables that are involved in clause $C_i$.
Finally, an $m$-ary $\mathrm{AND}$-gate accepts accepting as inputs all outputs of the $\mathrm{OR}$-gates.
Post-selecting this $\mathrm{AND}$-gate on $1$,
implements an ``effect'' and
represents the satisfiability requirement,
which is equivalent to postselecting all outputs of the $\mathrm{OR}$-gates to $1$, by the property of the $\mathrm{AND}$-gate.
A postselected-on-$1$ $\mathrm{OR}$-gate
is an $\mathrm{OR}$-constraint and
is a tensor with entries the OR truth table.
This tensor is represented by an all-negated $\mathrm{H}(0)$-box.
Diagrammatically we have:
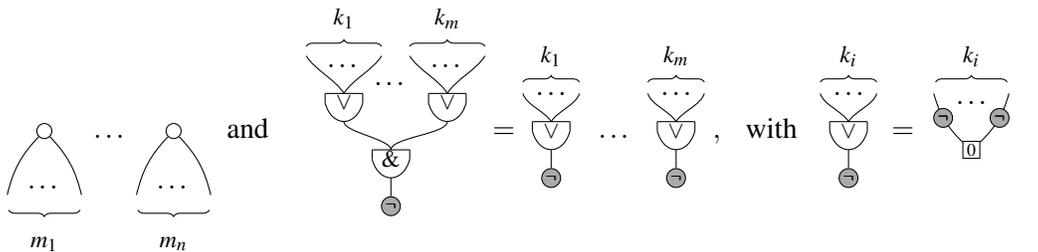
\vspace*{-1ex}
\[ \input{Boolvars-evensimpler1.tikz} ~\dots ~\input{Boolvars-evensimplern.tikz}
~~\mathrm{and}~~~~~ \input{AND-postsel.tikz} = \input{OR-box-postsel-gray1-k1.tikz}~~\dots~~\input{OR-box-postsel-gray1-km.tikz}~,~~~\mathrm{with}~~~\input{OR-box-postsel-gray1-ki.tikz} ~=~ \input{H-zero-noout.tikz}~~~.\]
Thus, any $\phi$ in CNF is expressed as a bipartite $\mathrm{ZH}$ tensor network
by connecting $\mathrm{Z}$-spiders
with all-negated $\mathrm{H}(0)$-boxes. We can capture the fact that a variable appears negated in a given clause by removing a negation from that wire, since two negations cancel out.

In this way, we can exactly capture the instance $\phi$ as a closed ZH-diagram, with one Z-spider for each variable, one $\mathrm{H}(0)$ box for each clause, and wires (possibly with NOTs) connecting them.
\end{proof}

Note that the $\mathrm{H}(0)$-box can be viewed as an $\mathrm{NAND}$-constraint, in the sense that its tensor elements are given by $(\mathrm{H}(0))_{x_1...x_k} = \mathrm{NAND}(x_1, \ldots, x_k)$.
Also, the postselected-circuit construction above returns the tensor network representations of CNFSAT of~\cite{10.21468/SciPostPhys.7.5.060}.

\begin{example}\label{ex:cnf-zh}
Consider the following CNF formula:
\begin{equation}\label{eq:zh-sat-example}
\phi = C_1\wedge C_2\wedge C_3\wedge C_4 = (\neg x_1 \vee x_2 \vee \neg x_3 \vee x_4)\wedge(x_1\vee\neg x_3\vee x_5)\wedge(\neg x_3\vee x_5\vee x_6)\wedge(x_2\vee \neg x_6)
\end{equation}
Its corresponding $\mathrm{ZH}$ tensor network is:
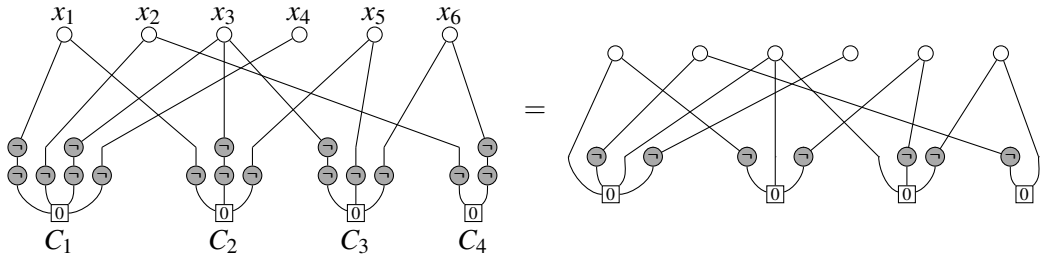
\[\input{zh-sat-example-formula.tikz}
\ \ =\ \ 
\input{zh-sat-example-formula-simp.tikz}\]
\end{example}


\subsection{Evaluating tensor networks with the ZH calculus}

The $\mathrm{ZH}$ calculus is not just a notation for tensor networks, but a \emph{sound} system 
of diagrammatic transformations
that preserve \emph{semantics}, i.e. the tensor represented by the network.

The most basic diagram transformations, or ``rewrites'', allowed in the ZH calculus 
are ways that sub-networks of ZH generators may be transformed while preserving semantics.
For $\text{Mat}[\mathbb{C}]$,
the $\mathrm{ZH}$ rewrite rules
are not only
\emph{sound},
but also
\emph{complete} --- meaning that any two networks using $\mathrm{ZH}$ generators which are equivalent, may be proven to be equivalent using \emph{only} diagram transformations.
In particular, any closed tensor network (e.g. the ZH-diagram representing $\#\phi$) expressed as a ZH diagram corresponds to a scalar (e.g. $[\#\phi]$). By completeness, it is always possible to transform a closed ZH-diagram into an efficient representation of that scalar, such as a disconnected collection of arity-0 generators.
We call \emph{full simplification} of a closed diagram a sequence of rewrites that removes all wires from the diagram.

Thus, as an alternative to tensor network contraction in the conventional way, we may in principle evaluate $[\#\phi]$ using transformations of ZH diagrams.
Taking this approach, the complexity of the computation is governed
by the size and complexity of the network as it is transformed as well as the number of transformation steps taken.
This is in contrast to full tensor contraction, in which the network becomes simpler (in terms of edge count) throughout the evaluation.
In this case, the complexity is in terms of the cost of storing the entries of the nodes which increases exponentially with their rank, which in turn is upper bounded by the number of incident wires.
One may elaborate on this by performing SVDs on the tensor nodes to limit the bond-dimension as much as possible \cite{levin2007a,gray2020,verstraete2004}.
Naturally, the different approach of the ZH calculus does not ensure that contraction can be performed efficiently --- the number of rewrites required to perform a full simplification may be exponential in the number of nodes of the input tensor, and the diagrams constructed in doing so may also in principle become exponentially complex --- but because the nodes range over a simple set, it is possible in some cases to use relationships between the generators and their types to identify special cases which may be evaluated efficiently.

In either case --- for full tensor contraction, or full simplification by ZH rewrites --- one generally expects to avoid mounting complexity in special cases, either to avoid mounting complexity of the tensor nodes (for tensor contraction) or mounting complexity of networks or difficulty in determining how to transform the networks (for ZH diagrams).
In the following sections, we consider the
complexity of ZH simplification strategies
for $\#\mathrm{SAT}$, $\#\mathrm{XORSAT}$, as well as their corresponding decision problems.

\vspace*{-1ex}
\section{Satisfiability with ZH rewrites}
\label{sec:decision}

One may ask whether or not the \emph{decision} problem $\mathrm{SAT}$ may similarly be represented by ZH diagrams.
The decision problem corresponds to determining whether or not the answer to a counting problem is non-zero.

\vspace*{-1ex}
\subsection{Matrices over semirings}
\label{sec:preliminaries}

The complex numbers $\mathbb C$, natural numbers $\mathbb N$ and the booleans $\mathbb B$ are all commutative semirings: sets $S$ equipped with two commutative and associative binary operations
$+: S \times S \to S$,
with identity element $0_S$,
and $\ast : S \times S \to S$,
with identity element $1_S$.
Operation $\ast$ distributes over $+$
and $0_S$ is absorbing for $\ast$
so that $\alpha \ast 0_S = 0_S \ast \alpha = 0_S, \forall a\in S$.
In $\mathbb{N}$, the operations $+$ and $\ast$ are just the usual addition and multiplication.
For $\mathbb{B}$,
we take $(+,0_{\mathbb B}) = (\vee, 0)$ and $(\ast,1_{\mathbb B}) = (\wedge,1)$, treating disjunction as addition and conjunction as multiplication.
We may also write $a \wedge b$ as a product $ab$, when $a,b \in \mathbb B$.

A homomorphism of semi-rings is a function $S \to T$ which preserves $0, 1, +$ and $*$. The most important semi-ring homomorphisms for our purposes are the inclusion $e : \mathbb N \to \mathbb C$ and the projection $p: \mathbb N \to \mathbb B$. In both cases, these maps are uniquely defined, due to the property of being homomorphisms.
Notably, we can see $\mathbb B$ as a quotient of $\mathbb N$, where we additionally impose the (seemingly nonsensical) equation $1 = 2$. In this case, $p$ is the quotient map.


For any semi-ring, we can define a monoidal category $\mathrm{Mat}[S]$ of matrices of $S$, whose objects are natural numbers, morphisms $\psi : m \to n$ are $n \times m$ matrices, composition is given by matrix multiplication, and $\otimes$ by tensor product. That is, for $\psi : m \to n$ and $\phi : m' \to n'$, $\psi\otimes \phi : mm' \to nn'$ is a matrix whose elements are $(\psi\otimes \phi)_{m'i + i'}^{n'j+j'} = \psi_i^j * \phi_{i'}^{j'}$. In particular, we can regard a morphism $\psi : m_1 \otimes \ldots \otimes m_k \to n_1 \otimes \ldots \otimes n_l$ as a $(k,l)$-tensor.

We can lift a semiring homomorphism $h: S \to T$ to a functor $H : \mathrm{Mat}[S] \to \mathrm{Mat}[T]$ simply by applying $h$ to each of the elements of a matrix. In this way, we can obtain a faithful functor $E : \mathrm{Mat}[\mathbb N] \to \mathrm{Mat}[\mathbb C]$ from $e$ and a full functor $P : \mathrm{Mat}[\mathbb N] \to \mathrm{Mat}[\mathbb B]$ from $p$.


\vspace*{-1ex}
\subsection{From counting to deciding by change of semiring}


Counting and deciding
may be interpreted as asking essentially the same question of ``counting'', but relative to different semirings:
either in $\mathbb N$ for the counting problem, or in $\mathbb B$ for the decision problem.
Therefore,
a closed $\mathrm{ZH}$ diagram
representing a $\phi$ in CNF evaluates to
$[\#\phi]$ by interpreting the diagram as a tensor network over $\mathbb{N}$
and to $[\phi]$
when the diagram is viewed as a tensor with coefficients in $\mathbb{B}$.
However, graph-partitioning based optimisers for the tensor contraction path are expected to exhibit similar performance regardless of the choice of semiring over which the tensor network's coefficients take values~\cite{10.21468/SciPostPhys.7.5.060}.
This is the case even for easy problem families such as $\mathrm{XORSAT}$ or $2\mathrm{SAT}$.
We present an approach to reasoning about boolean tensor networks using the ZH calculus which
allows us to recover the efficient solvability of said easy problems.

\vspace*{-1ex}
\subsection{Boolean tensor rewrites from the ZH calculus}
\label{sec:booleanFromZH}

Consider approaching the idea of rewriting boolean tensor diagrams, from the direction
of taking an existing diagrammatic calculus, which is well suited to describe matrices over $\mathbb{N}$,
such as the $\mathrm{ZH}$ calculus,
and ``projecting'' it down to the boolean semiring $\mathbb B$.
Subtly, this cannot be simply done by changing the semiring from which the node parameters are chosen.
This is because
there is no corresponding element to $-1$ in the booleans ($1$ has no additive inverse) whereas $\mathrm{H}(-1)$-boxes 
play a prominent role in the $\mathrm{ZH}$ rewrites.
Put another way, there is no semiring homomorphism $\mu: \mathbb C \to \mathbb B$ (or even $\nu : \mathbb Z \to \mathbb B$).
Despite this,
we can still reason soundly about boolean tensor networks with the $\mathrm{ZH}$ calculus.
However, while $\mathrm{ZH}$ is defined for generators taking parameters in $\mathbb C$,
among these generators are ones taking parameters in $\mathbb N$.
Diagrams which are composed exclusively of such generators will represent tensors over $\mathbb N$.
Furthermore, some gadgets of the ZH calculus --- in particular, the $\mathrm{X}$-spider and $\mathrm{NOT}$-dot
of Eqn.~\eqref{eq:grey-spider-and-neg-def} --- 
represent matrices over $\mathbb{N}$.


Any tensor network involving the $\mathbb{N}$-valued generators together with these gadgets will also represent tensors over $\mathbb N$.
We may then consider ways to reason about boolean tensors through the use of these diagrams.
Note that $\mathbb N$-valued ZH generators, together with NOT-dots, suffice to represent the ZH normal forms (and condensed normal forms) of any $\mathbb N$-valued tensor with indices of dimension $2$, as described in Eqn.~\eqref{eqn:ZHnf}.
This motivates the following definition:
\begin{definition}
  A \emph{NatZH diagram} is a $\mathrm{ZH}$ diagram generated from
  \vspace*{-1ex}
  \begin{equation*}
  \begin{aligned}[b]
      \Big\langle~~\input{Z-spider.tikz}~,
    \qquad
      \input{H-alpha-small.tikz}\!\!:\, \alpha \!\in\! \mathbb N~,
    \qquad
      \begin{tikzpicture}
	\begin{pgfonlayer}{nodelayer}
		\node [style=not] (0) at (0, 0) {};
		\node [style=none] (1) at (0, 1) {};
		\node [style=none] (2) at (0, -1) {};
	\end{pgfonlayer}
	\begin{pgfonlayer}{edgelayer}
		\draw (1.center) to (2.center);
	\end{pgfonlayer}
\end{tikzpicture}
~~\Big\rangle.
  \end{aligned}
  \end{equation*}
\end{definition}
\noindent
Some of the basic $\mathrm{ZH}$ rewrites
may be interpreted as rewrites of $\mathrm{NatZH}$ diagrams, but not all.
Notably, any rule involving an $\mathrm{H}(-1)$-box cannot be realised on $\mathrm{NatZH}$ diagrams, and as the definitions of gray nodes in Eqn.~\eqref{eq:grey-spider-and-neg-def} involve these, rewrites involving them also pose problems.
However, just as it is possible to reason about real polynomials by making use of the complex numbers, any theorem of $\mathrm{ZH}$ which describes the equivalence of $\mathrm{NatZH}$ diagrams is still usable to reason about $\mathrm{NatZH}$ diagrams, regardless of whether the \emph{intermediate} steps preserve the set of $\mathrm{NatZH}$ diagrams. This is essentially re-stating the fact that $\mathrm{Mat}[\mathbb N]$ embeds faithfully in the larger category $\mathrm{Mat}[\mathbb C]$.



Having defined a sub-theory of $\mathrm{ZH}$ which maps onto $\mathrm{Mat}[\mathbb N]$, we can approach the subject of boolean tensors via the functor $P: \mathrm{Mat}[\mathbb N] \to \mathrm{Mat}[\mathbb B]$ which projects $\mathbb N$-matrices down to $\mathbb B$-matrices. This projection will preserve any equation we proved between $\mathbb N$-matrices, e.g. those proven with the ZH-calculus. But also, more rules become true. Notably, since in booleans $1 = 2$, $\mathrm{H}(2)$-boxes are the same as  $\mathrm{H}(1)$-boxes, which will have a dramatic effect in Section~\ref{sec:2sat}.


\section{Identifying Efficiently Simplifiable Instances}

In this section we study the above introduced paradigmatic counting and decision problem families
whose complexity is well studied in the literature.
We will observe that
application of rewrites that
result in \emph{variable elimination},
or killing $\mathrm{Z}$-spiders,
is an efficient diagrammatic technique for solving the problem.
In contrast, (potentially) hard instances are those whose analogous rewrites either block further simplification or exponentially grow the diagram. We will illustrate precisely what we mean by this using the examples in this section.

\subsection{\#XORSAT and XORSAT}

We begin with the easy case of parity constraints.
Since $\mathrm{\#XORSAT}\in$P,
then it trivially follows that $\mathrm{XORSAT}\in$P.
Interestingly however,
this problem shows rich behaviour.
When one tries to anneal to a solution,
one encounters a glassy landscape which makes this approach inefficient~\cite{RicciTersenghi1639,patil2019obstacles}.

As in Proof.~\ref{pr:cnf2diag},
an $\mathrm{XOR}$-constraint is a postselected-on-$1$ $\mathrm{XOR}$-gate.
Since the $\mathrm{XOR}$-gate is the $\mathrm{X}$-spider,
a $\mathrm{XOR}$-constraint is a $\mathrm{NOT}$-spider.
A gray dot can absorb a number of $\mathrm{NOT}$-dots and by \emph{double negation elimination}
it is an $\mathrm{X}$-spider if it absorbs an even number of negations and a $\mathrm{NOT}$-spider if odd.
Also, a $\mathrm{NOT}$-dot gets copied by a white dot:
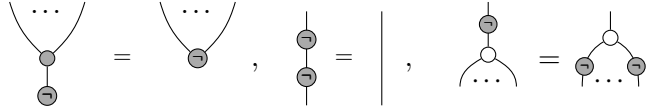
\vspace*{-1ex}
\[ \input{./figures/neg-spider.tikz}~~, ~~~~ \input{double-negation.tikz}~~~,~~~~~~
  \input{negcopy-rule-short-nodeco.tikz}
 \]
Therefore, any \#XORSAT instance diagram is represented in terms of:
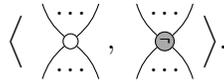
\vspace*{1ex}
\[\Big\langle~\input{Z-spider.tikz}~,~\input{X-spider-pi.tikz}~\Big\rangle.\]

Consider Example~\ref{ex:cnf-zh}, but with the $\mathrm{OR}$-constraints replaced by $\mathrm{XOR}$-constraints. Then one gets an example of a $\#\mathrm{XORSAT}$ formula and its representative $\mathrm{ZH}$ tensor network:
\begin{equation*}\label{eq:zh-xorsat-example}
\phi = (\neg x_1 \oplus x_2 \oplus \neg x_3 \oplus x_4)\wedge(x_1\oplus\neg x_3\oplus x_5)\wedge(\neg x_3\oplus x_5\oplus x_6)\wedge(x_2\oplus \neg x_6)
\end{equation*}
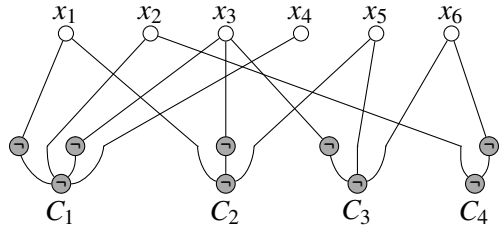
\[\input{zh-xorsat-example-formula.tikz}\]

The $\mathrm{X}$-spider,
and by consequence the $\mathrm{NOT}$-spider satisfy the $\mathrm{ZH}$ bialgebra law~\cite{wetering2019completeness}.
Also, a looping wire evaluates to $0$
if the wire is negated
and to $2$ if the wire is naked.
In the boolean case a naked looping wire can be trivially removed since $p(2)=1$:
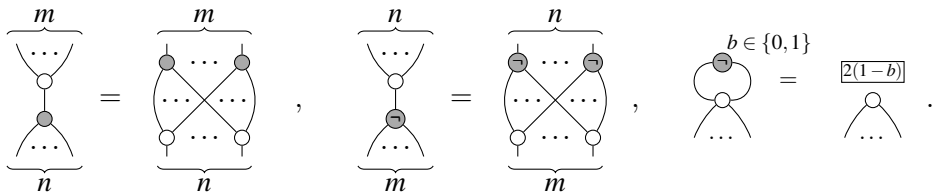
\vspace*{1ex}
\[\input{ZX-bialgebra.tikz}~~~~~, ~~~~~~~\input{ZX-bialgebra-neg.tikz}~~~ ,~~~~~~     \input{not-loop-rule.tikz}~~~.\]

The graph defining the problem is bipartite;
every wire connects a white dot
with a gray dot (either $\mathrm{X}$-spider or $\mathrm{NOT}$-spider).
After one application of the bialgebra rewrite rule,
we can always fuse dots of the same colour:
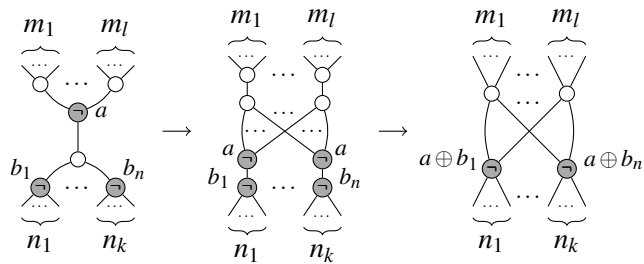
\[\input{XORSAT-spiderkill.tikz}\]
Iteratively using the bialgebra rewrite,
we can eliminate variables
by introducing only polynomially many edges.
In fact, the number of wires introduced
during this dance between gray and white spiders
is quadratically upper-bounded by
the number of nodes in the network.
In the case when the solution is not zero,
full simplification returns $[\#\phi]=2^c$, resulting from the number $c$ of disconnected components of the network after the iterated application of the bialgebra rule.
For the decision problem, exactly
the same procedure returns $[\phi]=p(2^c)=1$

\subsection{\#2SAT}

Now that we've seen that $\#\mathrm{XORSAT}$
can be efficiently solved by eliminating variables, we turn to the seemingly simple but actually hard $\#2\mathrm{SAT}$ problem
and attempt a similar strategy.
An instance
$\#\phi\in\#2\mathrm{SAT}$
is represented by a $\mathrm{NatZH}$ diagram generated by
variable tensors, \emph{binary} $\mathrm{H}(0)$-boxes, and negations. The latter two generators compose the binary $\mathrm{OR}$-constraint:
\begin{equation}\label{eq:2sat-gen}
\Big\langle~~\input{Z-spider.tikz}~~,~~~~~\begin{tikzpicture}
	\begin{pgfonlayer}{nodelayer}
		\node [style=none] (0) at (0, 1) {};
		\node [style=H box] (2) at (0, 0) {$0$};
		\node [style=none] (1) at (0, -1) {};
	\end{pgfonlayer}
	\begin{pgfonlayer}{edgelayer}
		\draw (0) to (1);
	\end{pgfonlayer}
\end{tikzpicture}
~~,~~~~~\begin{tikzpicture}
	\begin{pgfonlayer}{nodelayer}
		\node [style=not] (0) at (0, 0) {};
		\node [style=none] (1) at (0, 1) {};
		\node [style=none] (2) at (0, -1) {};
	\end{pgfonlayer}
	\begin{pgfonlayer}{edgelayer}
		\draw (1.center) to (2.center);
	\end{pgfonlayer}
\end{tikzpicture}
~~\Big\rangle ~~,~~~~~\input{2OR-box-postsel-gray1.tikz} ~=~ \input{H-zero-noout-2.tikz}
\end{equation}
We locally arrange the wires around a variable
so to separate those that go through a negation from the naked wires.
To eliminate a variable, we use the following theorem.

\begin{theorem}\label{thm:sharpsat-rule}
The following rule holds in the ZH-calculus, for all $m,n \geq 0$:
\begin{equation}\label{eq:sharp2sat-rule}
\input{sharp2sat-rule.tikz}
\end{equation}
\end{theorem}

\begin{proof}
By completeness of the ZH-calculus, it suffices to show the matrices of the LHS and RHS are equal. In each case, the matrix has a single 2 in the top-left corner, 1's in the rest of the first row and column, and 0's elsewhere.

We can evaluate matrix elements by pre- and post-composing with computational basis states $\{\begin{tikzpicture}
	\begin{pgfonlayer}{nodelayer}
		\node [style=gray dot] (0) at (0, -0.25) {};
		\node [style=none] (1) at (0, 0.5) {};
	\end{pgfonlayer}
	\begin{pgfonlayer}{edgelayer}
		\draw (0) to (1.center);
	\end{pgfonlayer}
\end{tikzpicture}
 = \ket 0,\ \begin{tikzpicture}
	\begin{pgfonlayer}{nodelayer}
		\node [style=not] (0) at (0, -0.25) {};
		\node [style=none] (1) at (0, 0.5) {};
	\end{pgfonlayer}
	\begin{pgfonlayer}{edgelayer}
		\draw (0) to (1.center);
	\end{pgfonlayer}
\end{tikzpicture}
 = \ket 1 \}$ and their adjoints and simplifying using the rules of the ZH-calculus. In particular, we use the following rules:
\begin{equation}\label{eq:tidy-up}
\input{basis-copy.tikz}~~~,
\qquad
\input{hbox-basis-0.tikz}~~~,
\qquad
\input{hbox-basis-1.tikz}~~~,
\qquad
\input{hbox-zero-basis.tikz}
\end{equation}
and split into 3 cases:
(i) If all the inputs are plugged with \begin{tikzpicture}
	\begin{pgfonlayer}{nodelayer}
		\node [style=gray dot] (0) at (0, -0.25) {};
		\node [style=none] (1) at (0, 0.5) {};
	\end{pgfonlayer}
	\begin{pgfonlayer}{edgelayer}
		\draw (0) to (1.center);
	\end{pgfonlayer}
\end{tikzpicture}
 and outputs with \begin{tikzpicture}
	\begin{pgfonlayer}{nodelayer}
		\node [style=gray dot] (0) at (0, 0.25) {};
		\node [style=none] (1) at (0, -0.5) {};
	\end{pgfonlayer}
	\begin{pgfonlayer}{edgelayer}
		\draw (0) to (1.center);
	\end{pgfonlayer}
\end{tikzpicture}
, the LHS simplifies to a single Z-spider with no inputs and outputs and the RHS to a single H-box labelled by $2$. In both cases, this equals the scalar $2$. (ii) If at least 1 input is \begin{tikzpicture}
	\begin{pgfonlayer}{nodelayer}
		\node [style=not] (0) at (0, -0.25) {};
		\node [style=none] (1) at (0, 0.5) {};
	\end{pgfonlayer}
	\begin{pgfonlayer}{edgelayer}
		\draw (0) to (1.center);
	\end{pgfonlayer}
\end{tikzpicture}
 and 1 output is \begin{tikzpicture}
	\begin{pgfonlayer}{nodelayer}
		\node [style=not] (0) at (0, 0.25) {};
		\node [style=none] (1) at (0, -0.5) {};
	\end{pgfonlayer}
	\begin{pgfonlayer}{edgelayer}
		\draw (0) to (1.center);
	\end{pgfonlayer}
\end{tikzpicture}
, the LHS simplifies to a diagram containing at least 1 copy of $\begin{tikzpicture}
	\begin{pgfonlayer}{nodelayer}
		\node [style=not] (0) at (0, 0.25) {};
		\node [style=none] (1) at (0, -0.5) {};
	\end{pgfonlayer}
	\begin{pgfonlayer}{edgelayer}
		\draw (0) to (1.center);
	\end{pgfonlayer}
\end{tikzpicture}
\circ\begin{tikzpicture}
	\begin{pgfonlayer}{nodelayer}
		\node [style=gray dot] (0) at (0, -0.25) {};
		\node [style=none] (1) at (0, 0.5) {};
	\end{pgfonlayer}
	\begin{pgfonlayer}{edgelayer}
		\draw (0) to (1.center);
	\end{pgfonlayer}
\end{tikzpicture}
$ and RHS contains an $H(0)$ box with no inputs and outputs. Hence both sides go to $0$. (iii) Otherwise the LHS and the RHS both simplify to many copies of $\begin{tikzpicture}
	\begin{pgfonlayer}{nodelayer}
		\node [style=gray dot] (0) at (0, 0.25) {};
		\node [style=none] (1) at (0, -0.5) {};
	\end{pgfonlayer}
	\begin{pgfonlayer}{edgelayer}
		\draw (0) to (1.center);
	\end{pgfonlayer}
\end{tikzpicture}
 \circ \begin{tikzpicture}
	\begin{pgfonlayer}{nodelayer}
		\node [style=white dot] (0) at (0, -0.25) {};
		\node [style=none] (1) at (0, 0.5) {};
	\end{pgfonlayer}
	\begin{pgfonlayer}{edgelayer}
		\draw (0) to (1.center);
	\end{pgfonlayer}
\end{tikzpicture}
$ or $\begin{tikzpicture}
	\begin{pgfonlayer}{nodelayer}
		\node [style=not] (0) at (0, 0.25) {};
		\node [style=none] (1) at (0, -0.5) {};
	\end{pgfonlayer}
	\begin{pgfonlayer}{edgelayer}
		\draw (0) to (1.center);
	\end{pgfonlayer}
\end{tikzpicture}
 \circ \begin{tikzpicture}
	\begin{pgfonlayer}{nodelayer}
		\node [style=white dot] (0) at (0, -0.25) {};
		\node [style=none] (1) at (0, 0.5) {};
	\end{pgfonlayer}
	\begin{pgfonlayer}{edgelayer}
		\draw (0) to (1.center);
	\end{pgfonlayer}
\end{tikzpicture}
$, which goes to $1$.
\end{proof}

In addition to the matrix derivation given in the proof above, equation~\eqref{eq:sharp2sat-rule} has a logical interpretation as well, in light of the representation of CNFs as ZH-diagrams from Theorem~\ref{thm:zh-sat-rep}. \textit{Propositional resolution}, i.e. reducing a pair of clauses of the form $(P \vee x)$ and $(\neg x \vee Q)$ to a clause $(P \vee Q)$, is a form of modus ponens which is particularly well-suited to CNFs. If we just focus on the $\mathrm{H}(0)$-boxes in equation~\eqref{eq:sharp2sat-rule}, we see that this transformation replaces all of the clauses containing a single variable $x$ with all of the possible resolutions not containing $x$:
\begin{equation}\label{eq:resolution}
\bigwedge_i (x \vee y_i) \wedge \bigwedge_j (\neg x \vee z_j)
\ \implies\ \ 
\bigwedge_{ij} (y_i \vee z_j)
\end{equation}
Indeed the conclusion of this implication is satisfiable if and only if its premise is.

However, a lingering question is: what is going on with the extra $\mathrm{H}(2)$ on the RHS of \eqref{eq:sharp2sat-rule}? This captures precisely the fact that, for every satisfying assignment to the conclusion of \eqref{eq:resolution} where all the $y_i, z_j$ are true, the premise has $2$ satisfying assignments, one where $x = 0$ and one where $x = 1$.


The extra $\mathrm{H}(2)$-box in the RHS of \eqref{eq:sharp2sat-rule} prevents us from being able to iteratively apply this rule to a CNF instance to obtain a solution in polynomial time. Indeed we should expect to find some obstruction to any polynomial-time strategy, since $\#2\mathrm{SAT}\in$ \#P-complete. However, we will now see what happens when we switch to a category where $2 = 1$.




\subsection{2SAT}\label{sec:2sat}

We now look at the
decision version of the above problem
which is in P.
The ZH-diagram of an instance $\phi\in\mathrm{2SAT}$ is generated just as in Theorem~\ref{thm:zh-sat-rep}, but it is evaluated as a tensor network over the Booleans, rather than the natural numbers. In particular, the parameter $2 = 1 + 1$ on the RHS of equation~\eqref{eq:sharp2sat-rule} becomes $1 \vee 1 = 1$ in the case of Booleans. So, the extra $\mathrm{H}(2)$-box on RHS vanishes:
\ctikzfig{2sat-simp}

Hence, for the booleans, we get a new, simpler rule:
\begin{equation}\label{eq:2sat-rule}
\input{2sat-rule.tikz}
\end{equation}



This rule no longer introduces H-boxes that would obstruct further application of the rule to its neighbours. Hence,
any 2SAT instance can be fully simplified as follows.
Every application of \eqref{eq:2sat-rule}
eliminates one variable,
and after every variable elimination
we perform
white-dot fusions,
as in the case of $\#\mathrm{XORSAT}$,
since again the network is bipartite,
every wire, either naked or negated,
connects a
$\mathrm{H}(0)$-box with a white-dot.
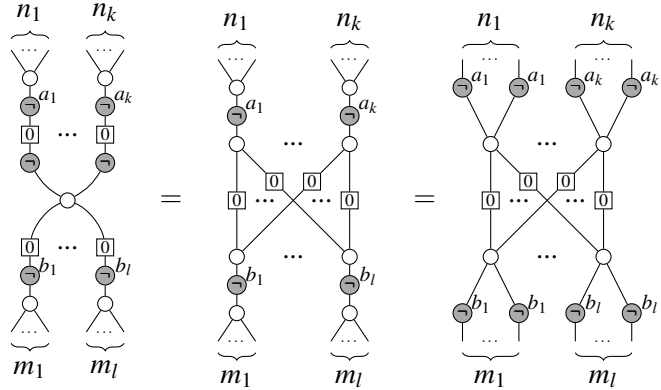
\[\input{2SAT-spiderkill.tikz}\]


We observe that all variables can be eliminated by iterative application of the
above boolean-spider killing rule \eqref{eq:2sat-rule}, followed by a small amount of `tidying up'. That is, we can remove `duplicate' $\mathrm{H}(0)$ boxes, which correspond to repeated clauses, and `self-loops', which correspond to a variable occurring twice in a clause, using the following rules:
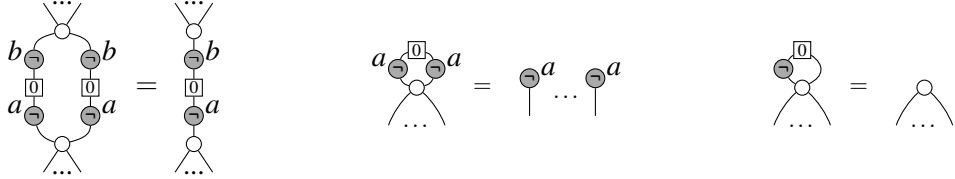
\[
\input{parallel-edge.tikz}~~~~~~~~~~~~~~~~~~~~~
\input{0box-loop-rules.tikz} \]
In the case of the first 2 kinds of loops, we end up with some basis states, so we are not a CNF-like ZH-diagram. However, applying the rules \eqref{eq:tidy-up} from the proof of Theorem~\ref{thm:sharpsat-rule} will either get back to CNF or produce a $0$ (i.e. `unsatisfiable').

With every white dot removed,
we introduce a number of wires
going through $\mathrm{H}(0)$-boxes
which is quadratic in the arity of that white dot. Since we can remove repeated $\mathrm{H}(0)$-boxes at every step, the arity of a given white dot is at most linear in the number of variables remaining.
We thus recover the polynomial cost
of full simplification of $2$SAT instances,
and we witness this purely in terms of diagram rewriting.

\subsection{\#SAT and SAT}

For the generic $\#\mathrm{SAT}$ case,
which is \#P-complete,
and the corresponding $\mathrm{SAT}$ case,
which is NP-complete,
the variable elimination rewrites
introduced above
for $\#2\mathrm{SAT}$ and $\mathrm{SAT}$
have analogues that
introduce $\mathrm{H}(0)$ boxes connected to more than just pairs of white dots. In particular, for \#SAT we have
that killing a white spider results in:
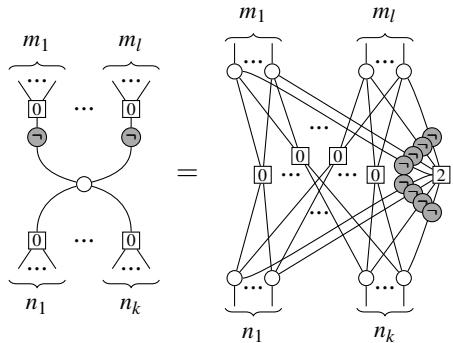
\[ \input{sharpsat-rule.tikz} \]
This rule is straightforward to prove from the \#2SAT version \eqref{eq:2sat-rule} and the basic ZH-calculus rules.

Again, passing to SAT, we can use $2 = 1$ to obtain a simpler rule:
\begin{equation}\label{eq:sat-rule}
\input{sat-rule.tikz}
\end{equation}
which we can indeed apply iteratively to solve a SAT instance. However, comparing the 2SAT simplification rule to the SAT version, we notice one crucial difference. In equation \eqref{eq:2sat-rule}, the arity of the $\mathrm{H}(0)$-boxes remains unchanged (at 2). Hence, we can argue that there are \emph{at most polynomially many} of them at each step. However, if the CNF contains clauses of 3 or more variables, iterating \eqref{eq:sat-rule} will cause the arities of $\mathrm{H}(0)$-boxes to grow arbitrarily large. Hence, in the worst case, we will obtain exponentially many distinct $\mathrm{H}(0)$-boxes, and the efficiency of our procedure is destroyed.


\section{Conclusions and Future Work}


We have imported rewrite rules from the sound and complete $\mathrm{ZH}$ graphical calculus in order to reason about tensor networks that represent instances of $\#\mathrm{SAT}$ and $\mathrm{SAT}$.
In this way,
we have identified that tensors of instances of problem families known to be in P,
such as $\#\mathrm{XORSAT}$ and $2\mathrm{SAT}$,
can be efficiently rewritten to a scalar encoding the solution to the problem.

In particular,
in order to be able to reason about
boolean tensor networks encoding
instances of $\mathrm{SAT}$,
we observed that tensors encoding
instances of $\#\mathrm{SAT}$
take values over the natural numbers.
The $\mathrm{ZH}$ calculus is complete for $\mathrm{Mat}[\mathbb{N}]$ since it is complete for $\mathrm{Mat}[\mathbb{C}]$,
and so for counting problems we simple used the available rewrites of $\mathrm{ZH}$.
In order to switch from counting to decision,
we interpreted decision as counting but over the boolean semiring.
Interestingly,
projecting from $\mathrm{Mat}[\mathbb{N}]$
to $\mathrm{Mat}[\mathbb{B}]$ gave rise to rewrite rules which made apparent
the difference in complexity between
$\#2\mathrm{SAT}$ and $2\mathrm{SAT}$.

For further work,
we aim to complete ongoing work
with preliminary results
on defining a sound and complete graphical calculus for \emph{boolean} tensor networks.
Such a calculus is motivated by the boolean projection $P$ we have used here
but is an \emph{independent} calculus from $\mathrm{ZH}$.
Notably,
there exists highly relevant work on the $\mathrm{ZX\&}$ calculus~\cite{colecomfort}, which provides a
sound and complete framework
for
natural-number tensors
as well as for boolean tensors.
It is inspired by the $\mathrm{ZX}$
calculus~\cite{Coecke2011} and adds the $\mathrm{AND}$ and $\mathrm{NOT}$ gates as primitive generators.

Finally, we note that the Python library PyZX~\cite{kissinger2019pyzx}, developed to support automated rewriting for ZX-diagrams supports the ZH-calculus as well. Building on this library, we aim to develop a useful tool both
for classical and quantum many-body problems.
It would be interesting to
benchmark such a library against
existing techniques for counting and decision problems,
as well as aid existing tensor contraction
methods by equipping them with rewrite strategies and
create hybrid algorithms for counting problems.

\section{Acknowledgments}

N.dB. is supported by a Fellowship funded by a gift from Tencent Holdings ({\tt tencent.com}).
A.K. would like to acknowledge support from AFOSR grant FA2386-18-1-4028.
K.M. is supported by an 1851 Research Fellowship ({\tt royalcommission1851.org}) and by Cambridge Quantum Computing Ltd..
K.M. acknowledges Stefanos Kourtis for inspiring discussions.

\appendix

\section{ZH rewrite rules}
\label{app:ZHrewrites}

Rewrite rules for the complete ZH-calculus in Fig.~\ref{fig:ZH-rules}.

\begin{figure}[h!]
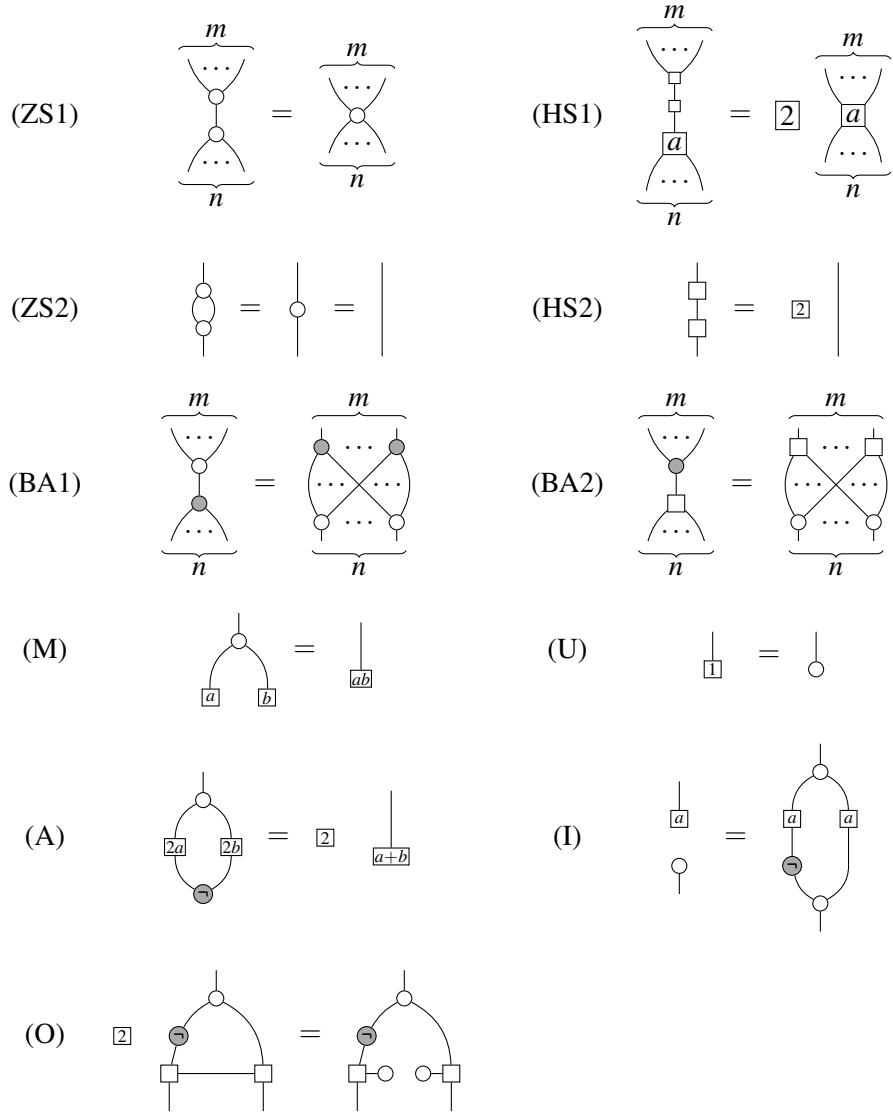

 \centering
 \begin{tabular}{ccccc}
  (ZS1) & \input{Z-spider-rule.tikz} & \qquad & (HS1) & \input{H-spider-rule.tikz} \\ &&&& \\
  (ZS2) & \input{Z-special.tikz} & & (HS2) & \input{H-identity.tikz} \\ &&&& \\
  (BA1) & \input{ZX-bialgebra.tikz} & & (BA2) & \input{ZH-bialgebra.tikz} \\ &&&& \\
  (M) & \input{multiply-rule.tikz} & & (U) & \input{unit-rule.tikz} \\ &&&& \\
  (A) & \input{average-rule.tikz} & & (I) & \input{intro-rule.tikz} \\ &&&& \\
  (O) & \input{ortho-rule.tikz} & & &
 \end{tabular}
 \caption{
 The rules of the ZH-calculus. $m,n$ are nonnegative integers and $a,b$ are arbitrary complex numbers.
 The right-hand sides of both \textit{bialgebra} rules (BA1) and (BA2) are complete bipartite graphs on $(m+n)$ vertices, with an additional input or output for each vertex.
 }
 \label{fig:ZH-rules}
\end{figure}

\bibliographystyle{eptcs}
\bibliography{refs}

\end{document}